\documentclass[prl,reprint,a4paper]{revtex4-1}
\usepackage{graphicx} % Include figure files
\usepackage{dcolumn}  % Align table columns on decimal point
\usepackage{amssymb,latexsym}
\usepackage{amsbsy} % for \boldsymbol
\usepackage{amsthm}
\usepackage{verbatim}
\usepackage{amsmath}
\usepackage{stmaryrd}
\usepackage{fullpage}
\usepackage{braket}

\usepackage[dvipsnames]{xcolor}

\newcommand{\Tr}{\operatorname{Tr}}

\renewcommand{\Im}{\operatorname{Im}}

\newcommand{\diag}{\operatorname{diag}}
\newcommand{\TC}{\operatorname{TC}}

\newcommand{\RR}{\mathbb{R}}

\newcommand{\calD}{\mathcal{D}}

\newcommand{\jp}{{\vec{j}_\text{p}}}
\renewcommand{\vec}[1]{\mathbf{#1}}
\newcommand{\rmd}{\mathrm{d}}

\newtheorem{theorem}{Theorem}

\newtheorem{corollary}{Corollary}

\begin{document}

\title{Lower semi-continuity of universal functional in paramagnetic current-density functional theory}
\author{Simen Kvaal}
\email{simen.kvaal@kjemi.uio.no}
\author{Andre Laestadius}
\author{Erik Tellgren}
\author{Trygve Helgaker}
\affiliation{Hylleraas Centre for Quantum Molecular Sciences, Department of Chemistry, University of Oslo, P.O. Box 1033 Blindern, N-0315 Oslo, Norway 
}%
\date{Friday 3 July 2020}

\begin{abstract}
    A cornerstone of current-density functional theory (CDFT) in its paramagnetic formulation is proven. After a brief outline of the mathematical structure of CDFT, the lower semi-continuity and expectation valuedness of the CDFT constrained-search functional is proven, meaning that there is always a minimizing density matrix in the CDFT constrained-search universal density functional. These results place the mathematical framework of CDFT on the same footing as that of standard DFT.
\end{abstract}

\maketitle

\section{Introduction}

Density-functional theory (DFT) is at present the most widely used tool for first-principles electronic-structure calculations in solid-state physics and quantum chemistry. DFT was put on a solid mathematical ground by Lieb in a landmark paper~\cite{Lieb1983} from 1983, where he introduced the universal density functional $F(\rho)$ as the convex conjugate to the concave ground-state energy $E(v)$ for an electronic system in the external scalar potential $v$.

For electronic systems under the influence of a classical external magnetic field $\mathbf{A}$, current-density functional theory (CDFT) was introduced by Vignale and Rasolt in 1987~\cite{Vignale1987}. In addition to the density $\rho$, the paramagnetic current density $\jp$ becomes a basic variable. The mathematical foundation of CDFT was put in place by Tellgren \emph{et al.}\,\cite{Tellgren2012} and Laestadius~\cite{Laestadius2014,Laestadius2014b} in the 2010s based on Lieb's treatment of the field-free standard case. However, 
a central piece of the puzzle has been missing---namely, whether the CDFT constrained-search functional $F(\rho,\jp)$ is lower semi-continuous and expectation valued~\cite{Kvaal2015}, i.e., that the infimum in its definition [see Eq.~\eqref{eq:cs} below] is in fact attained.

In this letter, we provide proofs of these assertions. The CDFT constrained-search functional is indeed convex lower-semicontinuous, and can therefore be identified with the CDFT Lieb functional---that is, a Legendre--Fenchel transform of the energy. \emph{Without} this fact, the ground-state energy functional $E(v,\mathbf{A})$ and the constrained-search functional $F(\rho,\jp)$ contain \emph{different} information. If $F(\rho,\jp)$ were not expectation valued, one would lose the interpretation of the universal functional as intrinsic energy, which is very useful in standard DFT.

For an $N$-electron system in sufficiently regular external potentials $v$ and $\mathbf{A}$, the ground-state energy is given by the Rayleigh--Ritz variation principle as
\begin{equation}
    E(v,\mathbf{A}) = \inf_{\Gamma} \Tr (\Gamma H(v,\mathbf{A}) ) , \label{eq:erg}
\end{equation}
where $H(v,\mathbf{A}) = T(\mathbf{A}) + W + \sum_{i=1}^N v(\vec{r}_i)$ is the electronic Hamiltonian with kinetic-energy operator $T(\mathbf{A}) = \tfrac{1}{2}\sum_{i=1}^N [-\mathrm{i}\nabla_i + \mathbf{A}(\mathbf{r}_i)]^2$ and two-electron repulsion operator $W$. The minimization is over all $N$-electron density matrices $\Gamma$ of finite kinetic energy, for which the one-electron density is $\rho \in X_{\text{L}} = L^1(\mathbb{R}^3)\cap L^3(\mathbb{R}^3)$, and $\jp \in X_\text{p} = \mathbf{L}^1(\mathbb{R}^3)\cap \mathbf{L}^{3/2}(\mathbb{R}^3)$~\cite{LAESTADIUS_JCTC15_4003}. (The boldface notation indicates a space of vector fields.) The external potential energy $(v \vert \rho) = \int_{\mathbb R^3} \!  v(\mathbf r)  \rho (\mathbf r)\,\mathrm d \mathbf r$, the paramagnetic and diamagnetic  terms $\tfrac{1}{2}(|\mathbf{A}|^2 \vert \rho)$ and $(\mathbf{A}\vert \jp) = \int_{\mathbb{R}^3} \! \mathbf{A}(\mathbf{r})\cdot \jp(\mathbf{r}) \, \mathrm{d} r$, and thus the Hamiltonian $H(v,\mathbf{A})$, are well defined for any  $v \in X_{\text{L}}' = L^{3/2}(\mathbb{R}^3) + L^\infty(\mathbb{R}^3)$ and $\mathbf{A} \in X_\text{p}' = \mathbf{L}^3(\mathbb{R}^3) + \mathbf{L}^\infty(\mathbb{R}^3)$, where $X^\prime_\text L$ and $X^\prime_\text p$ are the dual spaces of $X_{\text{L}}$ and $X_\text{p}$, respectively.  Examples of such potentials are the nuclear Coulomb potentials and uniform magnetic fields inside bounded domains. The symbol $X_\text{L}$ for the space of densities is so chosen to indicate it is the density space of  Lieb's analysis, while $X_\text{p}$ indicates ``paramagnetic'' current densities.

By a well-known reformulation of Eq.~\eqref{eq:erg}, we obtain the CDFT Hohenberg--Kohn variation principle 
\begin{equation}
    \begin{split}
        E(v,\mathbf{A}) &= \inf_{(\rho,\jp) \in X_\text L \times X_\text{p}} \big\{ F(\rho,\jp) \\  &\quad\qquad +(v + \tfrac{1}{2}|\mathbf{A}|^2 \vert \rho)  + (\mathbf{A}\vert \jp) \big\}.  \label{HK_cs}
    \end{split}
\end{equation}
Here the Vignale--Rasolt constrained-search density functional $F\colon X_\text L \times X_\text{p}  \to [0,+\infty]$ is defined by
%\begin{equation}
%    F[\rho] = \inf_{\Gamma\mapsto\rho} \Tr(\Gamma H_0 ) , \label{eq:cs}
%\end{equation}
\begin{equation}
    F(\rho,\jp) =  \inf_{\Gamma\mapsto (\rho,\jp)} \Tr(\Gamma H_0 ), 
\label{eq:cs}
\end{equation}
where $H_0 = T(\mathbf{0}) + W$ is the intrinsic electronic Hamiltonian, and $\Gamma \mapsto (\rho,\jp)$ means that the infimum is taken over all $N$-electron density matrices $\Gamma$ with density-current pair $(\rho,\jp) \in L^1(\mathbb{R}^3)\times \mathbf{L}^1(\mathbb{R}^3)$. Thus, if $(\rho,\jp)$ is not $N$-representable, we have $F(\rho,\jp)=+\infty$. 
The universal density functional $F$ is the central quantity in any flavor of DFT, whose
mathematical properties and approximation is of utmost importance to the field.

Although $E$ in Eq.\,\eqref{HK_cs} is not concave, it is readily seen that the reparametrized energy
\begin{equation}
\tilde{E}(u,\vec{A}) = E(u - \tfrac{1}{2}\vert \vec{A} \vert^2,\vec{A})
\end{equation}
is concave.  This reparametrization relies on a technical notion of {\it compatibility} of function spaces for the scalar and vector potentials~\cite{LAESTADIUS_JCTC15_4003}, satisfied for the potentials we consider here.

From the concavity and upper semi-continuity of the modified ground-state energy $\tilde{E}$, one deduces the existence of an alternative 
universal density functional $\hat{F} \colon X_{\text{L}} \times X_\text{p} \to [0,+\infty]$ related to the ground-state energy by Legendre--Fenchel transformations in the manner
\begin{align}
    \tilde{E}(u,\mathbf{A}) &= \inf_{(\rho,\jp)} \left\{ \hat{F}(\rho,\jp) + ( u \vert \rho) + (\mathbf{A}\vert \jp)  \right\}, \label{HK_cj}\\
    \hat{F}(\rho,\jp) &= \sup_{(u,\mathbf{A})} \left\{ \tilde{E}(u,\mathbf{A}) - (u \vert \rho) - (\mathbf{A}\vert\jp) \right\},
\end{align}
where the optimizations are over the space $X_\text{L}\times X_\text{p}$ and its dual $X_{\text{L}}'\times X_\text{p}'$, respectively.
As a Legendre--Fenchel transform, the functional $\hat{F}$ is convex and lower semi-continuous. 
In this formulation of CDFT, the ground-state energy $\tilde E$ and the universal density functional $\hat{F}$ contain precisely
\emph{the same information}: each functional can be obtained from the other and therefore contains all information about ground-state electronic systems in external scalar and vector fields.

From a comparison of the Hohenberg--Kohn variation principles in Eqs.\,\eqref{HK_cs} and~\eqref{HK_cj}, 
it is tempting to conclude that $\hat{F} = F$ 
%(when extended to all $X_\text L$ by setting  $F(\rho) = + \infty$ for $\rho \in X_\text L \setminus \mathcal I_N$)
are the same functional, producing the same ground-state energy for each $(v,\mathbf{A})$. 
However, there exist infinitely many functionals $\tilde F \colon X_\text L \times X_\text{p} \to [0,+\infty]$ that
give the correct ground-state energy $E(v,\mathbf{A})$ (but not necessarily the same minimizing density, if any) for each $(v,\mathbf{A})$ in the Hohenberg--Kohn variation principle. Each such $\tilde F$ is 
said to be an \emph{admissible density functional}~\cite{Kvaal2015}. Among these, the functional $\hat{F}$ stands out as being the only lower semi-continuous and convex 
universal density functional and a lower bound to all other admissible density functionals, $\hat{F} \leq \tilde F$. The functional $\hat{F}$, often called the closed convex hull of all admissible density functionals, is thus the most well-behaved admissible
density functional. Indeed, we may view it as a regularization of all admissible density functionals, known as the $\Gamma$-regularization in convex analysis. (This name is unrelated to our notation of density matrices.) 

A fundamental result of Lieb's analysis of DFT is the identification of the transparent   constrained-search density functional with the mathematically well-behaved closed convex hull $\hat{F}$. The identification follows since $F$ is convex and lower semi-continuous.
Whereas convexity follows easily for the CDFT Vignale--Rasolt functional $F$, the proof of lower semi-continuity is nontrivial. For standard DFT it is given in Ref.\,\onlinecite{Lieb1983}, and for CDFT in the present letter.

We simplify our analysis by merely assuming that the  density--current pairs are $(\rho,\jp) \in L^1(\mathbb{R}^3)\times \mathbf{L}^1(\mathbb{R}^3) = [L^1(\mathbb{R}^3)]^4$, which we denote as $X$. With this topology, the potentials must be taken to be bounded functions, $(v,\mathbf{A})\in X' = L^\infty(\mathbb{R}^3)\times \mathbf{L}^\infty(\mathbb{R}^3) = [L^\infty(\mathbb{R}^3)]^4$. 
This simplification is irrelevant in this context: if $F$ can be shown to be lower semi-continuous in the $[L^1(\mathbb{R}^3)]^4$ topology, it will be lower semi-continuous in any stronger topology,
as required if we enlarge the potential space to include more singular functions such as those in $X_\text{L}'\times X_\text{p}'$. Indeed, the original proof of lower semi-continuity of the standard DFT Levy--Lieb functional~\eqref{eq:cs} was with respect to the $L^1(\mathbb{R}^3)$ topology, from which the same property with respect to the $X_\text{L}$ topology immediately follows.

\section{Theorem and proof}

The intrinsic Hamiltonian $H_0 = H(0,\mathbf 0)$ is self-adjoint ($H_0 = H_0^\dag$) over $L^2_N$, the Hilbert space of square-integrable $N$-electron wavefunctions (with spin and permutational antisymmetry built in). The expectation values of $H_0$ and $H(v,\mathbf{A})$ are well-defined on the Sobolev space $H^1_N$, the subset of $L^2_N$ with finite kinetic energy. 

We denote by  $\mathcal{D}_N$ the convex set of $N$-electron mixed states with finite kinetic energy.  We have the mathematical characterization~\cite{Tellgren2014}
\begin{equation}
    \begin{split}
        \mathcal{D}_N = \big\{ \Gamma \in \TC(L^2_N) \mid \Gamma^\dag = \Gamma \geq 0, \, \Tr \Gamma = 1 &, \\
        \nabla_1 \Gamma \nabla_1^\dag \in \TC(L^2_N) & \big\},
    \end{split}
\end{equation}
where $\TC(L^2_N)$ is the set of \emph{trace-class operators} over $L^2_N$, the largest set of operators to which a basis-independent trace can be assigned. 
An operator $A$ is trace class if and only if the positive square root $|A| := \sqrt{A^\dag A}$ is trace class~\cite{Reed1980}. A self-adjoint operator $A$ is trace-class if and only if it has a spectral decomposition of the form $A = \sum_{k=1}^\infty \lambda_k \ket{\phi_k}\bra{\phi_k}$,
where $\{\phi_k\}$ forms an orthonormal basis and where $\sum_k \lambda_k$ is absolutely convergent. Now $A = \Gamma \in \mathcal{D}_N$ if and only if $\lambda_k \geq 0$, $\sum_k \lambda_k = 1$, $\{\phi_k\} \subset H^1_N$, and if the total kinetic energy is finite, $\sum_k \lambda_k \braket{\phi_k|T|\phi_k} < +\infty$.

For any $\psi \in H^1_N$, the density--current pair $(\rho,\jp) \in L^1(\mathbb{R}^3)\times \mathbf{L}^1(\mathbb{R}^3)$ is defined by 
%\green{
\begin{align}
    \rho(\mathbf{r}_1) &:=  N \int |\psi(\mathbf{r}_1;\tau_{-1})|^2 \, \rmd \tau_{-1}, \label{eq:rho-def}\\
    \jp(\mathbf{r}_1) &:= N \Im \int \psi^\ast(\mathbf{r}_1 ; \tau_{-1}) \nabla_1 \psi(\mathbf{r}_1 ; \tau_{-1}) \, \rmd \tau_{-1},
\end{align}
%}
where we integrate over all spin variables and over $N-1$ spatial coordinates, $\tau_{-1} = (\sigma_1,x_2,\cdots,x_N)$. For $A=\Gamma \in \mathcal D_N$, we can, for instance, compute $\rho= \rho_\Gamma$ from $\sum_k \lambda_k \rho_k$ with $\rho_k$ obtained from Eq.\,\eqref{eq:rho-def} with $\psi=\phi_k$ (and similarly for $\jp$).

The theorem involves the \emph{weak topology} on $X  = L^1(\mathbb{R}^3)\times \mathbf{L}^1(\mathbb{R}^3)$. Weak convergence of a sequence $\{x_n\} \subset X$, written $x_n \rightharpoonup x \in X$, means that, for any bounded linear functional $\omega \in X'$,  we have $\omega(x_n) \to \omega(x)$ as a sequence of numbers---that is, weak convergence is the pointwise convergence of all bounded linear functionals. 
Recall that the dual space of $L^1(\mathbb{R}^3)$ is $L^\infty(\mathbb{R}^3)$, 
so that $\rho_n \rightharpoonup \rho \in L^1(\mathbb{R}^3)$ if and only if 
$( f\vert\rho_n ) \to ( f \vert \rho)$ for every $f \in L^\infty(\mathbb{R}^3)$.
Likewise, $(\rho_n,\jp_n) \rightharpoonup (\rho,\jp)\in X$ if and only if 
$( f \vert\rho_n) \to (f \vert \rho)$ and
$( \mathbf{a}  \vert \jp_n)   \to (  \mathbf{a} \vert \jp)$
for every $(f,\mathbf{a}) \in X^\prime$. 
%for every $(f,\mathbf{f}) \in L^\infty(\mathbb{R}^3)\times \mathbf{L}^\infty(\mathbb{R}^3)$.
%\green{[Could it be okey not to display these last 3 eqs.? Also, maybe it is not so important, but should we have $f$ first and then the density variable?]}
% Given a space $L^2(M)$ of square-integrable functions over a measure space $M$, the \emph{Hilbert--Schmidt operators} $\HS(L^2(M))$ consists of $G \in K(L^2(M))$ such that $\|G\|_{\HS}^2 := \Tr (G^\dag G) < +\infty$, i.e., $G^\dag G \in \TC(L^2(M))$. Any $G \in \HS(L^2(M))$ can be associated with a unique Hilbert--Schmidt integral kernel $\mathcal{K}_g \in L^2(M\times M)$ and vice versa. Indeed, $\HS(L^2(M))$ is a Hilbert space with inner product
% \[ \braket{G_1,G_2}_{\TC} = \Tr(G_1^\dag G_2) = \iint_{M\times M} \mathcal{K}_{G_{1}}^*(x,y) \mathcal{K}_{G_{2}}(x,y) dxdy = \braket{\mathcal{K}_{G_{1}},\mathcal{K}_{G_{2}}}_{L^2(M\times M)}, \]
% so that $\HS(L^2(M))$ and $L^2(M\times M)$ are isometrically isomorphic. For any $G_1,G_2 \in \HS(L^2(M))$, we may define the diagonal of the product, $\diag(G_1^\dag G_2) = \int_M \overline{\mathcal{K}_{G_1}}(x,y) \mathcal{K}_{G_2} (x,y) dy \in L^1(M)$.  It is readily seen that we have the inclusion $\TC(L^2(M)) \subset \HS(L^2(M))$. 

The trace-class operators over a separable Hilbert space $\mathcal{H}$ are examples of \emph{compact operators}, an infinite dimensional generalization of finite-rank operators. Indeed, the set $K(\mathcal{H})$ of compact operators is the closure of the finite-rank operators in the norm topology and thus a Banach space. The dual space of $K(\mathcal{H})$ is in fact $\TC(\mathcal{H})$. For $B\in K(\mathcal{H})$ and $A \in \TC(\mathcal{H})$, the dual pairing is $ \Tr(BA)$. Similar to the weak topology for a Banach space, the dual of a Banach space  can be equipped with the \emph{weak-$*$ topology}. A sequence of trace-class operators $\{A_n\}$ converges weak-$*$ to $A \in \TC(\mathcal{H})$ if, for each $B \in K(\mathcal{H})$, $\Tr( B_n A) \to \Tr(BA)$.

We now state and prove our main result, from which lower semi-continuity follows in Corollary~\ref{cor}. The theorem is the CDFT analogue of Theorem~4.4 in Ref.\,\onlinecite{Lieb1983}.
\begin{theorem}\label{theorem:CDFT-wlsc-new}
  Suppose $(\rho,\jp) \in X$ and $\{(\rho_n,\jp_n)\}\subset X$ are such that  $F(\rho,\jp)<+\infty$ and $F(\rho_n,\jp_n)<+\infty$ for each $n \in \mathbb N$ and suppose further that $(\rho_n,\jp_n) \rightharpoonup (\rho,\jp)$. Then there exists $\Gamma\in\calD_N$ such that $\Gamma \mapsto (\rho,\jp)$ and $\Tr (H_0\Gamma) \leq \liminf_n F(\rho_n,\jp_n)$.
\end{theorem}
\begin{proof}[Proof of Theorem 1]
    The initial setup follows Ref.\,\onlinecite{Lieb1983}, which we here restate. Without loss of generality, we may replace $H_0 = T + W$ by $h^2 =
  T+W + 1$, which is self-adjoint and positive definite. The operator $h$ is taken to be the unique positive self-adjoint square root of $T + W + 1$.
  
Consider the sequence $\{g_n\}$ with elements $g_n := F(\rho_n,\jp_n)$. If $g_n\to+\infty$, then the statement of the theorem is trivially true. 
Assume therefore that $\{g_n\}$ is bounded. There then exists a subsequence such that $g := \lim_n g_n$ exists. 
%($\liminf_n g_n$ is the
%smallest accumulation point of the sequence $g_n$. $g$ is an
%accumulation point if and only if there is a subsequence which
%converges to $g$.)
Furthermore, for each $n$, there exists
$\Gamma_n\in\calD_N$ such that $\Gamma_n \mapsto (\rho_n,\jp_n)$ and $\Tr (h\Gamma_n h) = \Tr (h^2\Gamma_n) \leq g + 1/n$.  
To see this, select for each $n$ a density matrix
$\Gamma_n \mapsto (\rho_n,\jp_n)$ that satisfies $\Tr (h^2 \Gamma_n ) < g_n + 1/2n$ and
choose $m$ such that $\vert g - g_n \vert\ < 1/2n$ for
each $n > m$ (by taking a subsequence if necessary); for each $n  > m$, we then have
\begin{align}
0 \leq \Tr (h \Gamma_n h) - g &= \vert \Tr (h \Gamma_n h) - g \vert \nonumber \\ \leq \vert \Tr (h \Gamma_n &h) - g_n \vert+ \vert g_n - g \vert \leq 1 /n.
\end{align}
Using the sequence $\{h \Gamma_n h\}$, we next establish a candidate limit density operator $\Gamma \in \mathcal{D}_N$.
  
%% : First, for all $n$ there is a $\Gamma_n\in\calD_N$ such
%%   that $x_n = d(\Gamma_n)$, $\Tr \Gamma_n h^2 < g_n + 1/2n$. Second, choose
%%   $M$ so large that $|g-g_n|<1/2n$ for all $n>M$ (by taking a
%%   subsequence if necessary). Then, for all $n>M$,
%%   \[  | \Tr \Gamma_n h^2 - g | \leq |\Tr(\Gamma_n h^2) - g_n | + |
%%   g_n - g| < 1/n. \]
%%   Thus, $\Tr(\Gamma_n h^2) \leq g + 1/n$. Informally speaking, we choose
%%   our $\Gamma_n$ to  have energy closer to $F[\rho_n]=g_n$ the farther in
%%   the sequence we go. This is possible since $F[\rho_n]$ is an infimum. 

The dual-space sequence of (positive semi-definite) operators $y_n := h \Gamma_n h \in \TC(L^2_N)$ 
is uniformly bounded in the trace norm: $\|y_n\|_\text{TC} \leq g + 1$.  
By the Banach--Alaoglu theorem, %~\cite[Theorem IV.21]{Reed1980} states that 
a norm-closed ball of finite radius in the dual space 
is compact in the weak-$*$ topology.  Thus, there exists $y\in \TC(L^2_N)$ such that, for a subsequence, $\Tr (B y_n ) \rightarrow \Tr (B y)$ for each $B\in
K(L^2_N)$, meaning that $y$ is the (possibly nonunique) weak-$*$ limit of a subsequence of $\{ y_n \}$. 
The limit is positive definite, since the orthogonal projector $P_\Phi$ onto $\Phi\in L^2_N$ is a compact operator, which gives
  \begin{equation}
    \begin{split}
  \braket{\Phi|y|\Phi} & = \Tr (y P_\Phi) = \lim_n \Tr (y_n P_\Phi) \\
   &= \lim_n \braket{\Phi|y_n|\Phi} \geq 0.        
    \end{split}
\end{equation}
%   Moreover, since $y_n \geq 0$, 
%   \[ \Tr y_n = \sup \{\Tr P y_n\;:\; P \text{ proj.~onto finite dim subspace}\}.\]
%   The projectors $P$ are all compact operators, and thus $y \mapsto \Tr Py$ is continuous. A pointwise supremum of continuous functions is lower semicontinuous, which implies that
%   \[ \Tr |y| = \Tr y \leq \liminf_n \Tr y_n < +\infty. \]

%   We now demonstrate that $y\in \TC(\L^2_N)$. [I DON'T THINK WE NEED THIS, WE ONLY NEED TO SHOW THAT $y\geq 0$. TC FOLLOWS FROM WEAK-STAR. SIMEN]
%   Using the weak-$*$ convergence, we obtain $y\geq 0$, since for any $\Phi\in L^2_N$ with orthogonal projector $P_\Phi$ (a compact operator), 
%   \[ \braket{\Phi|y|\Phi} = \Tr (y P_\Phi) = \lim_n \Tr (y_n
%   P_\Phi) = \lim_n \braket{\Phi|y_n|\Phi} \geq 0,  \]
%   since $y_n\geq 0$ for all $n$.
%   Moreover, since $y_n \geq 0$, 
%   \[ \Tr y_n = \sup \{\Tr P y_n\;:\; P \text{ proj.~onto finite dim subspace}\}.\]
%   The projectors $P$ are all compact operators, and thus $y \mapsto \Tr Py$ is continuous. A pointwise supremum of continuous functions is lower semicontinuous, which implies that
%   \[ \Tr |y| = \Tr y \leq \liminf_n \Tr y_n < +\infty. \]
  
%  Using that the trace norm is lower semicontinuous,
%   \[ \Tr \Gamma = \Tr h^{-1} y h^{-1} \leq \liminf_n \Tr \Gamma_n = 1. \]

We now define $\Gamma = h^{-1}y h^{-1}$, which fulfills all the criteria for being an element of $\mathcal{D}_N$, except possibly $\Tr \Gamma = 1$, although $\Tr \Gamma \leq 1$ is already implied by the weak convergence. (Note that $\Gamma$ has finite kinetic energy since $\Tr (h^2 \Gamma) <+\infty$.) If we can show that $\Gamma \mapsto (\rho,\jp)$, then we are
done with the complete proof, since $\Gamma \in \mathcal D_N$ follows from
$\Tr\Gamma = N^{-1} \int_{\mathbb R^3} \rho(\mathbf r)\,\mathrm d \mathbf r = 1$ and since 
  \begin{equation}
    \begin{split}
        \Tr(h^2 \Gamma) &= \Tr y \leq \liminf_n \Tr y_n  \\
  &= \liminf_n \Tr (h^2\Gamma_n)  \\ &\leq \liminf_n \left\{F(\rho_n,\jp_n) + 1 /2n\right\}  \\
  &= \liminf_n F(\rho_n,\jp_n) .
    \end{split}
  \end{equation}  
%   \begin{equation}
%     \begin{split}
%         \Tr(H_0 \Gamma) &= \Tr y - 1 \leq \liminf_n \Tr y_n -   1  \\
%   &= \liminf_n \Tr (H_0\Gamma_n)  \\ &\leq \liminf_n \left(F(\rho_n,\jp_n) + 1 /2n\right)  \\
%   &= \liminf_n F(\rho_n,\jp_n) .
%     \end{split}
%   \end{equation}  

 Let $(\rho',\jp') \mapsfrom \Gamma$ be the density associated with $\Gamma$. To demonstrate that $(\rho',\jp')=(\rho,\jp)$, we recall that $(\rho_n,\mathbf{j}_{\text{p}n}) \rightharpoonup (\rho,\jp)$ by assumption. Since weak limits are unique, our proof is complete if we can 
 show that $(\rho_n,\mathbf{j}_{\text{p}n})\rightharpoonup (\rho',\jp')$ in $L^1(\mathbb{R}^3)\times\mathbf{L}^1(\mathbb{R}^3)$.
 The proof of  $\rho_n \rightharpoonup \rho'$ is given in Ref.\,\onlinecite{Lieb1983} and omitted here. 
 We here demonstrate that $\jp\rightharpoonup \jp'$ by showing that $(\mathbf{j}_{\text{p}n} - \jp'\,\vert\, \vec{a}) \to  0$ for each $\vec{a}\in \vec{L}^\infty(\RR^3)$.
 
%   \begin{align}
%    (\rho_n - \rho' \,\vert\,f) &\to 0 \quad  \forall f\in
%    L^\infty(\RR^3) \green{,}\label{eq:weak1} \\
%   (\jp_n - \jp')\,\vert\, \vec{f}) &\to  0\quad  \forall
%    \vec{f}\in \vec{L}^\infty(\RR^3). \label{eq:weak2}
%  \end{align}
 %for each $f\in L^\infty(\RR^3)$ and each $$\vec{f}\in \vec{L}^\infty(\RR^3)$.
 %\begin{align}
 %   \int_{\mathbb{R}^3}\!\! (\rho_n - \rho')f \,\rmd \vec r &\to 0 \quad  \forall f\in
 %   L^\infty(\RR^3) \green{,}\label{eq:weak1} \\
 %   \int_{\mathbb{R}^3}\!\! (\jp_n - \jp')\cdot \vec{f} \,\rmd \vec r &\to  0\quad  \forall
 %   \vec{f}\in \vec{L}^\infty(\RR^3). \label{eq:weak2}
 % \end{align}

%   For every $\varepsilon>0$ there is a bounded set $\Omega \subset \RR^3$ with
%   characteristic function $\chi$ such that for each component $\alpha\in\{1,2,3\}$ of the limit paramagnetic currents,
%   \[ \left|\int_{\mathbb{R}^3} \jp_\alpha (1-\chi) \right| < \varepsilon, \quad \left|\int_{\mathbb{R}^3} \jp'_\alpha(1-\chi)\right| <
%   \varepsilon. \] 
%   This follows since $\jp_\alpha$ is integrable.
%   Moreover, by the triangle inequality,  
%   \[ \left|\int_{\mathbb{R}^3} \jp_{n\alpha} (1-\chi)\right| < \left|\int_{\mathbb{R}^3} (\jp_{n\alpha} - \jp_\alpha) (1-\chi) \right| + \varepsilon.\]
%   Since $\jp_n \rightharpoonup \jp$, taking $n$ sufficiently large gives
%   \[ \left|\int_{\mathbb{R}^3} \jp_{n\alpha} (1-\chi)\right | < 2\varepsilon.\]

Let $\Omega\subset\mathbb{R}^3$ be a bounded domain with characteristic function $\chi$, equal to 1 on $\Omega$ and 0 elsewhere. 
Since $\rho,\rho'\in L^1(\mathbb{R}^3)$, we may, for a given $\varepsilon>0$, choose $\Omega$ sufficiently large so that  $\int (1-\chi) \rho\,\rmd \mathbf r  < \varepsilon$ 
and $ \int (1-\chi) \rho' \,\rmd \mathbf r< \varepsilon$. Since $\rho_n \rightharpoonup \rho$, we also have 
$\int (1-\chi)(\rho_n-\rho) \,\rmd \vec r + \varepsilon$ for sufficiently large $n$.
  From the triangle inequality, we obtain
  $\int (1-\chi) \rho_n \,\rmd \vec r \leq \int \!(1-\chi)(\rho_n-\rho) \,\rmd \vec r + \int (1-\chi) \rho' \,\rmd \mathbf r$, implying that
  $ \int (1-\chi) \rho_n \,\rmd \mathbf r < 2\varepsilon$ for sufficiently larger $n$.
  
In the notation $\tau = (\mathbf{r}_1,\tau_{-1}) = (x_1,x_2,\cdots,x_N)$ and $\tau_{-1} = (\sigma_1,x_2,\cdots,x_N)$ with space--spin coordinates $x_i = (\mathbf{r}_i,\sigma_i)$, let
$U_\alpha = N\Im \diag \partial_{1\alpha} \Gamma 
= N\Im \sum_\mu \lambda_\mu \overline{\psi_\mu(\tau)} \partial_{1\alpha} \psi_\mu(\mathbf{r}_1,\tau_{-1})$,
%   \begin{align*}
%       U_\alpha&= N\Im \diag \partial_{1\alpha} \Gamma \\
%      &= N\Im \sum_\mu \lambda_\mu \overline{\psi_\mu(X)} \partial_{1\alpha} \psi_\mu(\mathbf{r}_1,Y)
%   \end{align*}  
where $\alpha$ denotes a Cartesian component and
where we have introduced the spectral decomposition $\Gamma=\sum_\mu \lambda_\mu \ket{\psi_\mu}\bra{\psi_\mu} \in \mathcal D_N$ with $\psi_\mu \in H^1_N$.  We note that, if $\Gamma \mapsto (\rho_\Gamma,\mathbf{j}_{\text{p}\Gamma})$, then integration of $U_\alpha$ over $\tau_{-1}$ gives  $j_{\text{p}\alpha\Gamma}(\mathbf{r}) =  \int U_\alpha(\mathbf{r},\tau_{-1})\, \rmd \tau_{-1}$. 
  %More generally,
  %\[ N \int_{\mathbb{R}^{3N-3}} \rmd Y U_\alpha(x,Y) = 2 \partial_\alpha \rho_\Gamma(x) + i %\jp_{\Gamma,\alpha}(x) , \]
  %where we know a priori that $\partial_\alpha\rho_\Gamma \in L^1(\mathbb{R}^3)$ since $\rho_\Gamma$ is $N$-representable.]

% \green{SK: By this:
% Since $\rho\in L^1$, we may assume $\Omega$ is large so that
% \[ N \Tr((1-S)\Gamma) \leq \int \rho' (1-\chi) < \delta := \varepsilon^2/g. \]
% Furthermore, 
% \[ \int \rho_n (1-\chi) \leq \int (\rho_n - \rho) (1-\chi) + \delta.\]
% Since $\rho_n$ goes weakly to $\rho$, by taking $n$ large enough we may set
% \[ N \Tr((1-S)\Gamma_n) \leq \int \rho_n (1-\chi) < 2 \delta. \]
% We have that $\Tr(H_0\Gamma_n)$ is uniformly bounded by $g$. Thus,
% \begin{equation}
%   I(U_\alpha) \leq \varepsilon, \quad I(U_{\alpha,n}) \leq 2\varepsilon.  
% \end{equation}
% }

  Let now $S = \prod_{i=1}^N\chi(\mathbf{r}_i)$ be the characteristic function of $\Omega^N \subset \mathbb{R}^{3N}$.  By the definition of $U_\alpha$, we then have
\begin{align*}
I(U_\alpha)&:=\left\vert  \int (1-S) U_\alpha  \, \rmd \tau \right\vert \\
&\leq
N \!\! \int (1-S) \sum_\mu \lambda_\mu \vert  \psi_\mu  \vert    \vert \partial_{1,\alpha} \psi_\mu \vert \, \rmd \tau.
\end{align*}
Applying the Cauchy--Schwarz inequality twice, we obtain
\begin{align*}
I(U_\alpha)&\leq N \int (1-S)  \left(   \sum\nolimits_\mu  \lambda_\mu\vert  \psi_\mu  \vert^2\right)^{1/2}   \\ &\qquad\qquad\times
 \left(   \sum\nolimits_\mu \lambda_\mu  \vert \partial_{1,\alpha} \psi_\mu \vert^2 \right)^{1/2} \, \rmd \tau 
\\
&\leq( 2N)^{1/2} \left(    \int (1-S)        \sum\nolimits_\mu \lambda_\mu  \vert  \psi_\mu  \vert^2  \, \rmd \tau      \right)^{1/2}\!\!\!\!\!\!\! 
 \\ &\qquad\qquad\times\left( \frac{N}{2}\!\!  \int   \sum\nolimits_\mu \lambda_\mu  \vert \partial_{1,\alpha} \psi_\mu \vert^2 \, \rmd \tau \right)^{1/2}.
\end{align*}
Noting that $1-S \leq \sum_{i=1}^N [1-\chi(\mathbf{r}_i)]$ and using the symmetry of $|\psi_\mu|^2$, we obtain for the two factors
\begin{align*}
&\int (1-S)  \sum_\mu \lambda_\mu  \vert  \psi_\mu  \vert^2 \,  \rmd \tau  \leq  \int (1-\chi)\rho' \, \rmd \vec{r} < \varepsilon, \\
&\frac{N}{2}\!\! \int   \sum_\mu \lambda_\mu  \vert \partial_\alpha \psi_\mu \vert^2 \, \rmd \tau  = \Tr(T\Gamma)  \leq g .
\end{align*}
We conclude that $I(U_\alpha)^2 \leq 2N g \varepsilon$. Introducing 
%  \begin{align*}
$      U_{n,\alpha} =N\Im \diag \partial_{1\alpha} \Gamma_n$
%  \end{align*}
and proceeding in the same manner, we arrive at the bound $I(U_{n,\alpha})^2 \leq 4N g \varepsilon$, assuming that $n$ has been chosen so large that $ \int (1-\chi) \rho_n \,\rmd \mathbf r < 2\varepsilon$ holds.

We are now ready to consider the weak convergence $\mathbf{j}_{\text{p}n} \rightharpoonup \jp'$ in $\mathbf{L}^1(\mathbb{R}^3)$. For each  $\mathbf{a} \in \mathbf{L}^\infty(\mathbb{R}^3)$ and for sufficiently larger $n$, we obtain, using the Cauchy--Schwarz inequality and the H\"older inequality in combination with
the bounds $I(U_\alpha)^2 \leq 2N g \varepsilon$ and $I(U_{n,\alpha})^2 \leq 4N g \varepsilon$, the inequality
\begin{widetext}
\begin{equation}
    \begin{split}
        \left| \int (\mathbf{j}_{\text{p}n} - \jp')\cdot \mathbf{a} \,  \rmd \vec{r} \right| &\leq \sum_\alpha \left|\int
        (j_{\text{p}n\alpha}-j_{\text{p}\alpha}') a_\alpha \, \rmd \vec{r} \right|  = \sum_\alpha \left| \int (U_{n\alpha} - U_\alpha) a_\alpha(\mathbf{r}_1) \, \rmd \tau \right| \\
        & \leq  \sum_\alpha \left| \int (1-S) (U_{n\alpha} - U_\alpha) a_\alpha(\mathbf{r}_1)  \, \rmd \tau \right| + \sum_\alpha \left| \int S (U_{n\alpha} - U_\alpha) a_\alpha(\mathbf{r}_1) \, \rmd \tau \right| \\ 
        & \leq \sum_\alpha \|a_\alpha\|_\infty (6Ng\varepsilon)^{1/2} + \sum_\alpha \left| \int (U_{n\alpha} - U_\alpha) a_\alpha(\mathbf{r}_1) S\, \rmd \tau \right| .
    \end{split}
\end{equation}
\end{widetext}
% \green{SK: The following sequence of ineqs must be rewritten:}
% \begin{equation}
%     \begin{split}
%         \left| \int_{\mathbb{R}^{N}} (\jp_n - \jp')\cdot \mathbf{f} \right| &\leq \sum_\alpha \left|\int_{\mathbb{R}^{3}} (\jp_{n\alpha}-\jp'_{\alpha}) f_\alpha \right|  = \sum_\alpha \left| \int_{\mathbb{R}^{3N}} (U_{n\alpha} - U_\alpha) f_\alpha(x_1) \right| \\
%         & \leq  \sum_\alpha \left| \int_{\mathbb{R}^{3N}} (U_{n\alpha} - U_\alpha) f_\alpha(x_1) (1-S)\right| + \sum_\alpha \left| \int_{\mathbb{R}^{3N}} (U_{n\alpha} - U_\alpha) f_\alpha(x_1) S \right| \\ 
%         & \leq \sum_\alpha \|f_\alpha\|_\infty 3\varepsilon + \sum_\alpha \left| \int_{\mathbb{R}^{3N}} (U_{n\alpha} - U_\alpha) f_\alpha(x_1) S \right|
%     \end{split}
% \end{equation}
Since $\varepsilon>0$ is arbitrary, it only remains to show
$\int (U_{n\alpha} - U_\alpha) a_\alpha(\mathbf{r}_1) S \, \rmd  \tau \to 0$
as $n \to \infty$.  

Let $M$ be the compact multiplication operator associated with $a_\alpha(\mathbf{r}_1)S(\tau)$, a bounded function with compact support over $\mathbb{R}^{3N}$. Let $\Omega_\sigma = \{ \uparrow, \downarrow\}$ be the set consisting of the two spin states of the electrons. 
  We note that 
%\green{[I guess there is a reason why $(\Omega\times\{\uparrow,\downarrow\})^N$ is only written out here.]}
\begin{equation}
      \begin{split}
          \int U_{n\alpha}a_\alpha S \, \rmd \tau &=\int_{(\Omega\times\Omega_\sigma)^N} U_{n\alpha} a(\mathbf{r}_1) \, \rmd \tau \\
          &= N \Im \Tr (\partial_{1\alpha}\Gamma_n M) \\&= N \Im \Tr (h^{-1} M \partial_{1\alpha} h^{-1} y_n),
      \end{split}
  \end{equation}
 viewing $\Gamma_n$ as an  operator over $L^2((\Omega\times\Omega_\sigma)^N)$ by domain restriction of the spectral decomposition elements---that is, $\psi_\mu \in H^1((\Omega\times\Omega_\sigma)^N)$, meaning that the $2^N$ spin components of $\psi_\mu$ are in $H^1(\Omega^N)$. The spaces used here are not antisymmetrized, for simplicity.
  
 Our next task is to demonstrate that $B = h^{-1} M \partial_{1\alpha} h^{-1}$ is compact over $L^2((\Omega\times\Omega_\sigma)^N)$. We first show that $h^{-1}$ is compact  with range $H^1((\Omega\times\Omega_\sigma)^N)$. We have $h = \sqrt{T+W + 1}$ with domain $H^1_N((\Omega\times\Omega_\sigma)^N)$. Now $h^{-1}$ exists and is bounded since $-1$ is not in the spectrum of $T+W$---that is, $h^{-1} \colon L^2((\Omega\times\Omega_\sigma)^N) \to H^1((\Omega\times\Omega_\sigma)^N)$ is bounded. By the Rellich--Kondrachov theorem, $H^1(\Omega^N)$ (the standard Sobolev space without spin) is a compact subset of $L^2(\Omega^N)$. It follows that $H^1((\Omega\times\Omega_\sigma)^N)$ is a compact subset of $L^2((\Omega\times\Omega_\sigma)^N)$, since the tensor product of compact sets is compact. Hence, $h^{-1}$ is compact.
  
  Next, the operator $\partial_{1\alpha}$ is, by the definition of the Sobolev space $H^1(\Omega^N)$,  bounded from $H^1((\Omega\times\Omega_\sigma)^N)$ to $L^2((\Omega\times\Omega_\sigma)^N)$. Thus $\partial_{1\alpha} h^{-1}$ is bounded over $L^2((\Omega\times\Omega_\sigma)^N)$. It follows that $B \in K(L^2((\Omega\times\Omega_\sigma)^N))$ because it is a product of a compact operator $h^{-1}$ with a bounded operator $M\partial_{1\alpha}h^{-1}$.
  
From compactness of $B$, it follows that 
\begin{equation}
      \begin{split}
          \int U_{n\alpha} a_\alpha S \, \rmd \tau = N \Im &\Tr (B y_n) \\
          \to N \Im \Tr (B y) &= \int U_{\alpha} a_\alpha S \, \rmd \tau,
      \end{split}
  \end{equation}
by the weak-$*$ convergence of $y_n$ to $y$. We conclude that $\mathbf{j}_{\text{p}n}
  \rightharpoonup \jp'$ and hence that $(\rho,\jp)=(\rho',\jp')$, completing the proof.
\end{proof}
\begin{corollary}\label{cor}
  $F \colon L^1(\mathbb{R}^3)\times\vec{L}^1(\mathbb{R}^3) \to [0,+\infty]$ is lower semi-continuous and also weakly lower semi-continuous. 
\end{corollary}
\begin{proof}
Let $(\rho_n,\jp_n) \rightharpoonup (\rho,\jp) \in L^1(\mathbb{R}^3)\times \mathbf{L}^1(\mathbb{R}^3)$. From Theorem~\ref{theorem:CDFT-wlsc-new}, we then obtain
\begin{equation}
    F(\rho,\jp) \leq \Tr(H_0\Gamma) \leq \liminf_n F(\rho_n,\jp_n),
\end{equation}
where $\Gamma \mapsto (\rho,\jp)$. Hence, $F$ is weakly lower semi-continuous. By Mazur's Lemma~\cite{EkelandTemam}, weak lower semi-continuity of a convex function implies strong lower semi-continuity.
\end{proof}
\begin{corollary}\label{cor2}
    If $F(\rho,\jp)<+\infty$, then
    the infimum in the CDFT constrained-search functional is a minimum:
\begin{equation}
    F(\rho,\jp) = \min_{\Gamma\mapsto(\rho,\jp)} \Tr(\Gamma H_0).
\end{equation}
 \end{corollary}
\begin{proof}
    Simply take $(\rho_n,\jp_n) = (\rho,\jp)$ for all $n$, and apply Theorem~\ref{theorem:CDFT-wlsc-new}.
\end{proof}

\section{Conclusion}

We have extended Theorem~4.4 of Ref.\,\onlinecite{Lieb1983} to CDFT. As immediate corollaries, the constrained-search functional $F(\rho,\jp)$ is lower semi-continuous and expectation valued, that is, if $F(\rho,\jp)<+\infty$, then there exists a $\Gamma \mapsto (\rho,\jp)$ such that $F(\rho,\jp) = \Tr (H_0 \Gamma)$. These mathematical results are the final pieces in the puzzle of placing CDFT on a solid mathematical ground in a similar manner as done by Lieb for standard  DFT.

\section{Acknowledgments}
This work has received funding from the Research Council of Norway (RCN) under
CoE Grant Nos.~287906 and 262695 (Hylleraas Centre for Quantum Molecular Sciences) and from ERC-STG-2014 under grant agreement No 639508. 

\bibliographystyle{unsrt}
\bibliography{refs.bib}

\end{document}